\documentclass[12pt]{article}
\newtheorem{theorem}{Theorem}
\newtheorem{lemma}{Lemma}

\newenvironment{proof}{{\bf Proof.}}{\hfill\rule{2mm}{2mm}}

\newtheorem{remarka}{Remark}

\usepackage{array}
\newlength{\cellwid}
\makeatletter
 {\endarray}

\def\emline#1#2#3#4#5#6{%
       \put(#1,#2){\special{em:moveto}}%
       \put(#4,#5){\special{em:lineto}}}
\def\newpic#1{}

\title{\bf An approximation algorithm for the total cover problem}

\author{
{\bf Pooya Hatami }
\\[1mm]
{\small\it Department of Mathematical Sciences} \\
Sharif University of Technology \\
{\small e-mail: p\_hatami@ce.sharif.edu}}
\date{}
\begin{document}
\maketitle

\begin{abstract}
We introduce a $2$-approximation algorithm for the minimum total
covering number problem.
\end{abstract}

\section{Introduction}        

A {\sf vertex cover} of an undirected graph $G=(V,E)$ is a subset
$S\subseteq V$  such that if $e=uv\in E$, then $\{ u,v \} \cap S
\neq \emptyset$. A set $D \subseteq V \cup E$ is called a {\sf
total cover} if every element of $(V \cup E) \setminus D$ is
adjacent or incident to an element in $D$.

The notion of total covering is first defined in \cite{ABLN}, and
then studied in many papers~\cite{ALWZ,EM,Meir,PS}. Many
variations of the covering problems including vertex covers,
total covers, dominating sets, {\it et cetera} have been studied
previously (see~\cite{HHS}).

The minimum total cover problem was first shown to be NP-hard in
general graphs by Majumdar~\cite{Majumdar}, where he also gives a
linear-time algorithm for trees. Hedetniemi {\it et al.}
\cite{HHLMM} showed that the problem is NP-hard for bipartite and
chordal graphs. Manlove \cite{Manlove} demonstrates NP-hardness for
planar bipartite graphs of maximum degree $4$.

Trivially for every graph a vertex cover together with all
isolated vertices constitute a total cover. It is well-known that
a maximal matching can be used to find a vertex cover of size at
most twice the minimum vertex cover: If $M$ is a maximal matching
of the graph $G$, the set $S$ of all $2|M|$ vertices involved in
$M$ constitute a vertex cover of $G$. Moreover a vertex cover of
$G$ has at least $|M|$ elements, because every vertex is involved
in at most one matching edge. Thus taking the vertices which are
involved in a maximal matching gives a $2$-approximation
algorithm for the minimum vertex cover problem.

It is widely believed that it is {\sf NP}-hard to approximate the
vertex cover problem to within any factor smaller than $2$, and
recently Khot and Regev~\cite{Khot} proved that the {\sf Unique
Games Conjecture} would imply this. So far, the best known lower
bound is a recent result of Dinur and Safra~\cite{DS} which shows
that it is {\sf NP}-hard to approximate this problem to within any
factor smaller than $10 \sqrt{5}-21 \approx 1.36067$.

The approximability of the problem of finding a minimum total cover
does not seem to have received explicit attention in the literature
previously. However given a graph $G=(V,E)$, the relationship
$\alpha_2(G)=\gamma(T(G))$ holds, where $\alpha_2(G)$ denotes the
minimum size of a total cover in $G$, $\gamma(G)$ denotes the
minimum size of a dominating set in $G$, and $T(G)$ denotes the
total graph of $G$ (this is the graph with vertex set $V \cup E$,
and two vertices are adjacent in $T(G)$ if and only if the
corresponding elements are adjacent or incident as vertices or edges
of $G$). It follows from the correspondence that the minimum total
cover problem is approximable within a factor of $1+\log n$, where
$n=|V|$ \cite{Johnson}. Also, if $\Delta(G)\leq k$, then
$\Delta(T(G))\leq 2k+1$. It follows that, in a graph of maximum
degree $k$, the problem of finding a minimum total cover is
approximable within a factor of $H_{2(k+1)}-\frac{1}{2}$ \cite{DF},
where $H_i=\sum_{j=1}^i \frac{1}{j}$ is the $i$th Harmonic number.

We introduce a simple and elementary algorithm which finds a total
cover of size at most twice the size of an optimal total covering.
Note that, for $k\geq 3$, $H_{2(k+1)}-\frac{1}{2}\ge 2$, implying
that the above derived results would be improved upon this
$2$-approximation algorithm.


\section{The approximation algorithm}

In this section we introduce an approximation algorithm for
computing the minimum total cover number of a graph.

After that straightforward $2$-approximation algorithm for the
minimum vertex cover problem, it is tempting to try the same
algorithm for the total cover problem. It is easy to see that if
we modify this algorithm to include all isolated vertices too, we
obtain an approximation algorithm for the total cover problem. The
following example shows that the algorithm is not a
$(4-\epsilon)$-approximation: Consider the graph illustrated in
Figure~\ref{approx} for even $n$. The maximum matching of this
graph is of size $n$ while the set $S$ which consists of $v$ and
all edges of the form $e=u_i u_{i+1}$ is a total cover of size
$\frac{n}{2}+1$ of the graph. Lemma~1 will immediately conclude
that the mentioned algorithm has factor $4$.

\begin{figure}[ht]

\begin{center}



\unitlength=.40mm \special{em:linewidth 0.4pt}
\linethickness{0.4pt}

\begin{picture}(20.50,75.00)(10,5)

\put(-50.00,20.00){\circle*{3.00}}\
\put(-50.00,50.00){\circle*{3.00}}\
\put(-30.00,20.00){\circle*{3.00}}\
\put(-30.00,50.00){\circle*{3.00}}\
\put(-10.00,20.00){\circle*{3.00}}\
\put(-10.00,50.00){\circle*{3.00}}\
\put(10.00,20.00){\circle*{3.00}}\
\put(10.00,50.00){\circle*{3.00}}\
\put(30.00,35.00){\circle*{1.00}}\
\put(35.00,35.00){\circle*{1.00}}\
\put(40.00,35.00){\circle*{1.00}}\
\put(60.00,20.00){\circle*{3.00}}\
\put(60.00,50.00){\circle*{3.00}}\
\put(80.00,20.00){\circle*{3.00}}\
\put(80.00,50.00){\circle*{3.00}}\
\put(35.00,80.00){\circle*{3.00}}\
\emline{-50.00}{20.00}{1}{-50.00}{50.00}{2}
\emline{-30.00}{20.00}{1}{-30.00}{50.00}{2}
\emline{-10.00}{20.00}{1}{-10.00}{50.00}{2}
\emline{10.00}{20.00}{1}{10.00}{50.00}{2}
\emline{60.00}{20.00}{1}{60.00}{50.00}{2}
\emline{80.00}{20.00}{1}{80.00}{50.00}{2}

\emline{35.00}{80.00}{1}{-50.00}{50.00}{2}
\emline{35.00}{80.00}{1}{-30.00}{50.00}{2}
\emline{35.00}{80.00}{1}{-10.00}{50.00}{2}
\emline{35.00}{80.00}{1}{10.00}{50.00}{2}
\emline{35.00}{80.00}{1}{60.00}{50.00}{2}
\emline{35.00}{80.00}{1}{80.00}{50.00}{2}
\emline{-50.00}{20.00}{1}{-30.00}{20.00}{2}
\emline{-10.00}{20.00}{1}{10.00}{20.00}{2}
\emline{60.00}{20.00}{1}{80.00}{20.00}{2}
\put(-50.00,15.00){\makebox(0,0)[cc]{$u_1$}}
\put(-30.00,15.00){\makebox(0,0)[cc]{$u_2$}}
\put(-10.00,15.00){\makebox(0,0)[cc]{$u_3$}}
\put(10.00,15.00){\makebox(0,0)[cc]{$u_4$}}
\put(60.00,15.00){\makebox(0,0)[cc]{$u_{n-1}$}}
\put(80.00,15.00){\makebox(0,0)[cc]{$u_n$}}
\put(33.00,85.00){\makebox(0,0)[cc]{$v$}}
\put(-55.00,45.00){\makebox(0,0)[cc]{$v_1$}}
\put(-35.00,45.00){\makebox(0,0)[cc]{$v_2$}}
\put(-15.00,45.00){\makebox(0,0)[cc]{$v_3$}}
\put(5.00,45.00){\makebox(0,0)[cc]{$v_4$}}
\put(55.00,45.00){\makebox(0,0)[cc]{$v_{n-1}$}}
\put(75.00,45.00){\makebox(0,0)[cc]{$v_n$}}

\end{picture}

\caption{\label{approx} A hard example for maximal matching
algorithm.}

\end{center}

\end{figure}
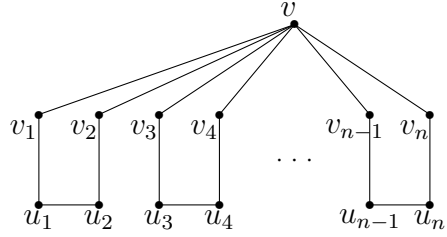

Next we introduce a $2$-approximation algorithm for this problem.
Consider a graph $G=(V,E)$ with $t$ isolated vertices. Let $M$ be
a maximum matching in $G$ of size $m$. Let $k$ be the number of
vertices that {\bf (i)} are not involved in $M$, and {\bf (ii)}
are adjacent to both endpoints of an edge in $M$; we call these
{\sf bad} vertices. Note that since $M$ is of maximum size, if a
bad vertex $w$ is adjacent to both endpoints of $e=uv$, then
neither $u$ nor $v$ is adjacent to any other vertex outside $M$.
We will find a total cover $S$ of size $m+k+t$ in $G$ through the
following algorithm:

\begin{enumerate}

\item Obviously every isolated vertex must be in $S$. Remove all
these vertices from $G$.

\item Select an edge $e=v_1 v_2$ from $M$ where both $v_1$ and
$v_2$ are adjacent to a bad vertex $v$. Add $v$ and $e$ to $S$ and
remove $v, v_1, v_2$ from $G$. Repeat this until all bad vertices
are removed.

\item After the first two steps the size of $S$ is $2k+t$. Now we
have a graph $G_1$ with a maximum matching $M_1$ of size $m-k$
without any bad vertices. Next we apply the following step:
\begin{itemize}

\item  Pick the edges $e=uv \in M_1$ in an arbitrary order, and
note that at most one of $u$ and $v$ is adjacent to some vertex in
$G_1\setminus M_1$:
\begin{itemize}
\item  If one of $u$ and $v$ is adjacent to some vertex $z \in
G_1\setminus M_1$ which is not covered by $S$, then add that
vertex to $S$.

\item  Otherwise add $e$ to $S$.

\end{itemize}

\end{itemize}
\end{enumerate}
It is clear that $S$ is of size $m+k+t$. We now show that $S$
covers every edge between vertices covered by $M_1$ (It is clear
that $S$ covers all other elements of $G$). Suppose that
$e_1=u_{1}v_{1}$ and $e_2=u_{2}v_{2}$ are matching edges in $M_1$
and the edge $e=u_{1}u_{2}$ is not covered. Then none of $e_1$,
$e_2$, $u_1$, $u_2$ are in $S$, and so both $v_1$, $v_2$ are in
$S$. Suppose that among $e_1$ and $e_2$, the edge $e_1$ was picked
first. So there is a vertex $w_1\in G_1 \setminus M_1$ adjacent to
$v_1$ and there is a vertex $w_2\in G_1 \setminus M_1$ which is
adjacent to $v_2$ but not to $v_1$. The path
$w_1$,$v_1$,$u_1$,$u_2$,$v_2$,$w_2$ is an augmenting path for
$M_1$ and this contradicts the fact that $M$ is a maximum matching
of $G$.

\begin{lemma}
The minimum total cover of $G$ has at least $\frac{m+k}{2}+t$
elements.
\end{lemma}
\begin{proof}
Call every triangle consisting of a  bad vertex $v$ and a matching
edge whose both endpoints are adjacent to $v$ a {\sf bad}
triangle. There are $k$ bad vertices, and no two bad vertices can
share a common matching edge, thus there exist at least $k$ bad
triangles.

Suppose that $S$ is a total cover in $G$. Let $A \subseteq S$ be
a \emph{maximal} set of \emph{edges} which covers $2|A|$ edges of
$M$, consisting of the edges each covering precisely two edges of
$M$. Let $B=S \setminus A$. Every bad triangle has at least one
edge which is not covered by $A$. Since no edge is incident to
two bad triangles with distinct vertices, no element (that is,
neither a vertex nor an edge) can cover two edges from two
disjoint bad triangles. Since the number of the bad triangles is
at least $k$, we have $|B|\geq k+t$, as $B$ has to cover the
isolated vertices too.

Since $B$ covers at most $|B|-t$ edges of $M$, $S=A \cup B$ covers
at most $2|A|+|B|-t$ edges of $M$. Thus $2|A|+|B|-t\geq m$. From
this inequality and $|B|\ge k+t$ we get $2(|A|+|B|)\ge m+k+2t$
which implies that $|S|=|A|+|B| \ge \frac{m+k}{2}+t$.
\end{proof}

\begin{theorem}
The minimum total cover problem admits a $2$-approximation
algorithm.
\end{theorem}
\begin{proof}
Immediately from Lemma~1.
\end{proof}

Consider the graph illustrated in Figure~\ref{approx}. Our
algorithm finds a total cover of size $n+1$. The graph in
Figure~\ref{approx} has a total cover of size $\frac{n}{2}+1$.
The result of our algorithm is $2-o(1)$ times the size of the
minimum total cover of the graph. So our algorithm is not a
$(2-\epsilon)$-approximation, for any $\epsilon$.

\bibliographystyle{plain}
\bibliography{tcover-sub}

\end{document}